\documentclass[a4paper, 12pt]{article}

\usepackage[sort&compress]{natbib}
\bibpunct{(}{)}{;}{a}{}{,} 

\usepackage{amsthm, amsmath, amssymb, mathrsfs, multirow, url, subfigure}
\usepackage{graphicx} 
\usepackage{ifthen} 
\usepackage{amsfonts}
\usepackage[usenames]{color}
\usepackage{fullpage}
\usepackage{multirow}
\usepackage{multicol}
\usepackage{booktabs}

\usepackage{algorithm}
\usepackage{algpseudocode}

\theoremstyle{plain} 
\newtheorem{thm}{Theorem}

\theoremstyle{definition}

\theoremstyle{remark}

\newcommand{\RR}{\mathbb{R}}

\newcommand{\N}{\mathsf{N}}

\newcommand{\E}{\mathsf{E}}
\newcommand{\V}{\mathsf{V}}
\newcommand{\prob}{\mathsf{P}}

\newcommand{\eps}{\varepsilon}

\newcommand{\iid}{\overset{\text{\tiny iid}}{\sim}}

\newcommand{\nm}{\mathsf{N}}

\newcommand{\biglp}{\Big{(}}
\newcommand{\bigrp}{\Big{)}}
\newcommand{\bigls}{\Big{[}}
\newcommand{\bigrs}{\Big{]}}
\newcommand{\biglb}{\Big{\{}}
\newcommand{\bigrb}{\Big{\}}}

\title{Variational empirical Bayes variable selection in high-dimensional logistic regression}
\author{Yiqi Tang\footnote{Department of Statistics, Colby College; {\tt ytang@colby.edu}} 
\quad and \quad Ryan Martin\footnote{Department of Statistics, North Carolina State University; {\tt rgmarti3@ncsu.edu}}
}

\date{\today}

\begin{document}

\maketitle 

\begin{abstract}
Logistic regression involving high-dimensional covariates is a practically important problem.  Often the goal is variable selection, i.e., determining which few of the many covariates are associated with the binary response.  Unfortunately, the usual Bayesian computations can be quite challenging and expensive.  Here we start with a recently proposed empirical Bayes solution, with strong theoretical convergence properties, and develop a novel and computationally efficient variational approximation thereof.  One such novelty is that we develop this approximation directly for the marginal distribution on the model space, rather than on the regression coefficients themselves.  We demonstrate the method's strong performance in simulations, and prove that our variational approximation inherits the strong selection consistency property satisfied by the posterior distribution that it is approximating. 
\smallskip

\emph{Keywords and phrases:} data-driven prior; generalized linear model; model selection consistency; posterior concentration; variational approximation.
\end{abstract}

\section{Introduction}
\label{S:intro}

Generalized linear models \citep[GLMs,][]{mccullagh.nelder.1983} are among the most widely used statistical models in applications.  Notable special cases include the standard Gaussian linear model for real-valued response variables and the logistic regression model for binary response variables.  The present paper focuses on the latter logistic regression model, which 
assumes that the binary $y_1, y_2, \ldots, y_n$ are independently distributed observations with mass function 
\[ f_\beta(y_i \mid x_i) = \exp\bigl[ y_i x_i^\top\beta - \log\{1+\exp(x_i^\top\beta)\} \bigr], \quad i=1,\ldots,n, \] 
where $x_i$ denotes row $i$ of the (assumed fixed) design matrix $X$, and $\beta$ is a $p\times 1$ vector of coefficients.  Software, including the {\tt glm} function in R, is available to fit such a model to observed data and return hypothesis tests and confidence intervals based on the standard asymptotic normality of the maximum likelihood estimators (MLEs).  

Recent research efforts have focused on high-dimensional settings in which the number of predictor variables, $p$, far exceeds the number of observations, $n$. Compared to the classical low-dimensional case, with $p < n$, if $p \gg n$, then it is not possible to fit such a model adequately using maximum likelihood alone: for one reason or another, regularization is required. Frequentist methods such as lasso \citep{tibshirani1996} and ridge regression \citep{hoerl.kennard.1970} add a penalty term to achieve this regularization. Many have investigated the performance of such methods in high-dimensional linear regression \citep[e.g.,][]{zhang.huang.2008}, and even in high-dimensional generalized linear models such as logistic regression \citep[e.g.,][]{salehi2019impact, wu2009genome, meier2008group}. On the Bayesian side, the standard Markov chain Monte Carlo (MCMC) methods employed in low-dimensional cases do not adequately scale to high-dimensional cases.  
For this and other reasons, there is a genuine interest in developing alternatives to the Bayesian's prior--model--MCMC pipeline in logistic regression and elsewhere.  These alternatives might involve a choice of prior that simplifies computations, it might involve alternatives to MCMC-based computation, or a combination of both.


On the choice-of-prior front, there are mainly two types of priors used in high-dimensional analysis, continuous shrinkage priors such as the horseshoe \citep[e.g.,][]{carvalho.polson.scott.2010, piironen2017sparsity} and spike-and-slab priors \citep[e.g.,][]{castillo.vaart.2012, belitser2020empirical}. Unfortunately, those priors with good concentration rate properties have heavier tails \citep[e.g.,][]{jeong2021posterior} and the corresponding posterior computations are more challenging and expensive, while thin-tailed conjugate Gaussian priors offer simplified posterior computations but have sub-optimal theoretical convergence properties \citep{castillo.vaart.2012}. For this reason, \citet{martin.mess.walker.eb} propose the use of data to center the prior, noting that if the prior is properly centered, then the tails would not matter much. With a strategic centering, the thin-tailed conjugate Gaussian prior can be used, simplifying the derivations without compromising on the computation efficiency. They use an empirically-driven conjugate prior, then combine it with the likelihood in almost the usual Bayesian way to form a data-driven posterior distribution. They show that this empirical prior can achieve the optimal minimax posterior concentration rate in high-dimensional linear regression. The simplicity and flexibility of this idea has been shown in a number of different applications, including monotone density estimation \citep{martin2019empirical}, high-dimensional Gaussian graphical models \citep{eb.gwishart}, piecewise polynomial sequence models \citep{ebpiece}, and particularly relevant to our work here, in high-dimensional generalized linear models \citep{tang2024empirical, lee.chae.martin.glm}. General theoretical results for this brand of empirical Bayes posteriors are presented in \citet{martin2019data}.

On the alternative-to-MCMC front, variational inference is now quite standard \citep[e.g.,][]{blei.etal.vb.2017, dhaka2021challenges, zhang2020convergence}. Instead of integration via Monte Carlo, the problem is converted to one that only requires optimization, and allows for fast computation without sacrificing (too much) on the theoretical properties. With the aforementioned empirical priors approach, variational inference can also be employed to aid in computations. Previous work along these lines in \citet{yang2020variational} focused on high-dimensional Gaussian linear regression.  In this paper, we focus on high-dimensional logistic regression, and present a novel variational approximation to the empirical priors posterior of \citet{tang2024empirical}, solving the problem with an optimization algorithm rather than MCMC. 

Variational inference in high-dimensional logistic regression itself is not new \citep[e.g.,][]{ray2020spike, zhang2019novel, jaakkola2000bayesian, guoqiang2022variational}. The novelty of our proposed method is that, instead of a more traditional Gaussian or Laplace prior, ours uses an empirically-centered prior, and the particular form gives us access to a relatively simple expression for the marginal posterior mass function of the configuration $S$, the active set that correspond to the active coefficients in the $\beta$ vector.  This expression, unfortunately, does not directly allow for inference on $S$, so we propose a variational approximation directly on this marginal posterior.  Rather than the relatively complicated variational approximation with the mean-field family---independent mixtures of Gaussian and point-mass distributions---as is common in the literature, we propose a simple independent-Bernoulli approximation to the marginal posterior for $S$, which yields a much simpler and transparent approximation.  This transparency allows us to easily establish, in Theorem~\ref{thm:vb.consistent} below, that the proposed variational approximation shares the same strong selection consistency property as the marginal posterior for $S$ that it is approximating.  There are results of this type for high-dimensional linear regression \citep[e.g.,][]{huang2016variational, guoqiang2022variational, ormerod.etal.2017}, and also some general results for other brands of variational approximations \citep{ohn2024adaptive}, but we are not aware of any results in the literature on selection consistency for variational approximations in high-dimensional logistic regression. Our simulations show that this method performs well compared to other existing methods, both Bayesian and frequentist, and that this method produces results that approximate the inclusion probabilities from the MCMC method detailed in \citet{tang2024empirical} well. Our method can easily and efficiently accommodate the large $n$ and $p$ settings that are not feasible with MCMC. 

The remainder of the paper is organized as follows. Section~\ref{S:background} briefly reviews variational inference and the empirical prior/posterior construction and properties as presented in \citet{tang2024empirical} and further studied in \citet{lee.chae.martin.glm}.  In Section~\ref{S:VIpast}, we give an overview of the variational inference framework commonly used in high-dimensional linear regression. Then, in Section~\ref{S:eb}, we introduce the setup of the problem and our empirically-centered prior and its associated posterior. In Section~\ref{S:VI}, we propose our novel variational approximation for said posterior, and in Section~\ref{S:CAVI} detail the coordinate ascent variational inference algorithm for computation and its derivation. Numerical results comparing our method with others are presented in Section~\ref{S:chap4results}, including comparisons with the EB-MCMC method from \citet{tang2024empirical}. Finally, Section~\ref{S:chap4discuss} offers a discussion of our method as well as directions for further work.

\section{Background}
\label{S:background}

\subsection{Variational inference}
\label{S:VIpast}

We focus our review here on the high-dimensional linear regression context. The key ideas and principles of variational inference are roughly the same for other problems.  Write $y = X\beta + \eps$, where $y$ is a $n\times 1$ vector of observed data, $X$ is a $n\times p$ matrix, $\beta$ is a $p\times 1$ vector of coefficients, and $\eps$ is a $n\times 1$ vector of error. Re-express the vector $\beta$ as a pair $(S, \beta_S)$, where $S\subset\{1, 2, \ldots, p\}$ is the set of indices that correspond to active signals and $\beta_S$ consists of the non-zero values corresponding to the configuration $S$.  

Express $S$ as a binary $p$-vector, where $S_j=1$ if variable $j$ is in $S$ and $S_j=0$ otherwise. Following \citet{ray2022variational}, write the variational approximation as 
\begin{equation}
\label{eq:q.theta}
q_\theta(S, \beta_S) = \prod_{j=1}^p q_{j, \theta}(\beta_j \mid S_j) \, q_{j,\theta}(S_j),
\end{equation}
where
\begin{align*}
q_{j,\theta}(S_j) & = 
        \begin{cases}
        \phi_j, & S_j=1\\
        1-\phi_j, & S_j=0
        \end{cases} \\
q_{j, \theta}(\beta_j \mid S_j) & = 
        \begin{cases}
            \N(\beta_j \mid \mu_j, \tau_j^2), & S_j=1\\
            \delta_0(\beta_j), & S_j=0.
        \end{cases}
\end{align*}
Here, the variational parameter $\theta$ consists of three $p$-vectors, $\theta=(\mu, \tau^2, \phi)$. The $\beta_j$'s are taken to be independently distributed from a mixture of Gaussian and point-mass, i.e., $\beta_j\sim \phi_j\N(\mu_j, \tau^2_j) + (1-\phi_j)\delta_0$. Collecting all the $\beta_j$'s together, we have a family of approximate densities 
\[ \mathcal{L} = \Bigl\{ \bigotimes_{j=1}^p
\{\phi_j\N(\mu_j, \tau^2_j) + (1-\phi_j)\delta_0 \}: \mu_j\in \RR, \, \tau^2_j >0, \, \phi_j\in[0,1] \Bigr\}.\]
This is known as a mean-field family, commonly used in variational inference \citep{blei.etal.vb.2017}. It allows for a relatively accurate approximation that is still easy to compute, thanks to the independence of the $\beta_j$'s. 

If $\pi^n$ is the joint posterior for $(S,\beta_S)$ to be approximated, and $q_\theta$ is the approximation family in \eqref{eq:q.theta}, then the {\em evidence lower bound} (ELBO) is defined as 
\[K(\theta) = \E_{(S, \beta_S)\sim q_\theta}\log\{{\tilde\pi}^n(S, \beta_S)/q_\theta(S, \beta_S)\},\]
where $\tilde\pi^n$ is the un-normalized version of $\pi^n$. With an optimization algorithm such as coordinate or gradient ascent, the optimal variational parameter is obtained by $\hat{\theta} = \arg \max_\theta K(\theta)$ and the original posterior distribution can now be approximated with $q_{\hat{\theta}}$.

This approximation has a number of drawbacks. The variational parameter $\theta$ is of dimension $3p$, which is a huge space to optimize over. The algorithm needs to approximate over a large number of dimensions, and as one could imagine, the more dimensions there are, the more difficult it is to obtain an accurate approximation. This is the primary motivation behind our proposed variational approximation on $S$ itself, yielding a much simpler approximation with only a $p$-vector $\phi$ in $[0,1]^p$ as our variational parameter.

\subsection{Empirical priors and logistic regression}
\label{S:eb}


Here we review the empirical prior and posterior construction for logistic regression as proposed in \citet{tang2024empirical}.  First, we introduce a bit more notation.  

As above, any $p$-vector $\beta$ can be expressed as $(S,\beta_S)$, the model structure and structure-specific coefficients.  For a particular structure $S$, write $X_S$ for the submatrix determined by keeping only the columns of $X$ whose column index is contained in $S$. Then, the log-likelihood function at $\beta \equiv (S,\beta_S)$ is $\ell_n(S,\beta_S) = \sum_{i=1}^n y_ix_{i,S}^\top\beta_S - \log\{1+\exp(x_{i,S}^\top\beta_S)\}$, where $x_{i,S}$ is the $i^\text{th}$ row (expressed as a column vector) of the matrix $X_S$.  For a given $S$, the MLE $\hat\beta_S$ of $\beta_S$ is found by solving the likelihood equation, $\dot\ell_n(S,\beta_S)=0$.  Then the observed Fisher information is $J_n(S) = -\ddot\ell_n(S,\hat\beta_S) = X_S^\top W(S,\hat\beta_S) X_S$, where $W(\cdot,\cdot)$ is a diagonal matrix with entries 
\[ W_{ii}(S,\beta_S) = \frac{\exp(x_{i,S}^\top \beta_S)}{\{1+\exp(x_{i,S}^\top \beta_S)\}^2}, \quad i=1,\ldots,n. \] 

Now we are ready to introduce our empirical prior for $\beta=(S,\beta_S)$, and we do so hierarchically: first a marginal prior for $S$ and then a conditional prior for $\beta_S$, given $S$.  The marginal prior for $S$ is informative and does not depend on the data; we simply use the complexity prior from \citet{castillo.vaart.2012}, i.e., 
\[ \pi_n(S) \propto \textstyle\binom{p}{|S|}^{-1} \, p^{-a|S|}, \]
where $a > 0$ is a hyperparameter.  Next, the data-dependent conditional prior is
\[ (\beta_S \mid S) \sim \nm_{|S|}\bigl( \hat\beta_S, \gamma J_n(S)^{-1} \bigr), \]
where $\gamma > 0$ is another hyperparameter.  Note first that the prior is centered at the $S$-specific MLE, and, second, that if $\gamma$ is constant/bounded, then the variance is $O(n^{-1})$.  Intuitively, there can be no benefit to the empirically-driven center if the spread is not relatively tight, and it has been shown in \citet{martin.mess.walker.eb}, \citet{tang2024empirical}, and elsewhere that this is the amount of spread needed for good posterior concentration properties.  The prior on the complexity $|S|$ is rather strong so, in practice it is beneficial to choose a small value for $a$.  Similarly, with the data-driven choice of center, it is advantageous to choose a small value for $\gamma$ so that the prior for $\beta_S$ is rather tightly concentrated.  In our examples below, we take $a=0.01$ and $\gamma = 0.1$. Putting the two pieces together, we can write the full empirical prior for the $p$-vector $\beta$ as 
\[ \Pi_n(d\beta) = \sum_S \pi_n(S) \, \nm_{|S|}\bigl(d\beta_S \mid \hat{\beta}_S, \gamma J_n(S)^{-1} \bigr) \times \delta_0(d\beta_{S^c}), \quad \beta \in \RR^p. \]

Following \citet{martin.mess.walker.eb}, the posterior distribution combines the likelihood and prior in almost the usual Bayesian way, i.e., 
\[ \Pi^n(d\beta) \propto L_n(\beta)^\alpha \Pi_n(d\beta), \quad \beta \in \RR^p, \]
where $\alpha \in (0,1)$ is a constant to be specified. Since data is used in both the prior and the likelihood, the use of $\alpha$ allows us to slightly discount the likelihood portion, which helps us steer clear of any problematic double-use of data. In practice, we take $\alpha$ to be very close to 1 ($\alpha=0.99$), so our posterior does not differ too much from a genuine Bayesian posterior with $\alpha=1$. The use of a discount factor $\alpha$ is now quite common in the (generalized) Bayes literature \citep[e.g.,][]{grunwald2012, jeong2021posterior, miller.dunson.power, walker.hjort.2001, syring.martin.scaling}.

In principle, we can integrate out the $\beta_S$ to obtain a marginal posterior on $S$. Unlike the linear model case in \citet{martin.mess.walker.eb}, however, this marginal posterior for $S$ is not available in closed-form, so we employ Laplace's method to get the approximation
\begin{equation}
    \pi^n(S)\propto \pi_n(S) (1+\alpha\gamma)^{-|S|/2}L_n(S, \hat{\beta_S})^\alpha.
\label{eq:laplace}
\end{equation}
For more details on these derivations, see \citet{tang2024empirical}. Laplace approximation is not always accurate, but has been shown to work well in certain configurations \citep[e.g.,][]{barber2016laplace, shun1995laplace}. Fortunately, we only need the approximation to be accurate for sparse configurations, i.e., small $|S|$ spaces, made precise in Appendix A.2 of \citet[][App.~A.2]{tang2024empirical} and \citet[][Sec.~5.1]{lee.chae.martin.glm}. 

With this approximation, we can now use Metropolis--Hastings to (approximately) sample directly from the posterior for $S$, obtaining inclusion probabilities for each variable for variable selection.  We define the inclusion probability for variable $j$ as 
\begin{equation}
\label{eq:ip}
\pi^n(S_j=1) := \pi^n(\{S: S \ni j\}) \approx \frac{1}{M} \sum_{m=1}^M 1\{S^{(m)} \ni j\}, 
\end{equation}
the proportion of sampled configurations---the $S^{(m)}$'s---that include variable $j$.  This is precisely the strategy utilized in previous work. Specifically, the Metropolis--Hastings algorithm uses a symmetric proposal distribution that allows the algorithm to move in the $S$ space from state $S$ to a new $S'$ that is different from $S$ in only one position of $S$. This works well with moderate $p$, but becomes increasingly inefficient as $p$ grows, hence our motivation to develop a more efficient strategy.

\section{A new variational approximation}
\label{S:VI}

\subsection{Construction}

Thanks to Laplace approximation, our marginal posterior for $S$ has a very nice form, at least up to normalization. We take advantage of the simple expression in Equation~\eqref{eq:laplace} by applying our variational approximation directly to this marginal posterior. That is, instead of what is typically done in the literature, where the variational approximation is applied to the joint posterior for $(S, \beta_S)$, we propose to focus directly on a variational approximation of the marginal posterior for $S$, simplifying the problem. 

Recall that the configuration $S$ can be interpreted as both a subset of $\{1,2,\ldots,p\}$ and a binary $p$-vector $S=(S_1,\ldots,S_p)$ with $S_j=1$ indicating that the variable $j$ is active and $S_j=0$ indicating that it is inactive.  Taking the latter interpretation, a natural choice of approximation to the marginal posterior of $S$ is an independent Bernoulli model.  That is, our proposed approximate marginal mass function for $S$ is 
\[q_\phi(S) = \prod_{j=1}^p \phi_j^{S_j} (1-\phi_j)^{1-S_j},\]
where $\phi = (\phi_1, \phi_2, ..., \phi_p)^\top \in [0,1]^p$ is the variational parameter to be determined.  The goal is to find the value $\phi$ that minimizes the Kullback--Leibler divergence of $\pi^n$ from $q_\phi$.  In other words, we aim to minimize the objective function 
\begin{equation}
K(\phi) = \E_{S\sim q_\phi} \Bigl\{ \log\frac{q_\phi(S)}{\pi^n(S)} \Bigr\}.
    \label{eq:KL}
\end{equation}
For our theoretical investigations below, we will keep the full marginal posterior $\pi^n$ in the denominator, but for our numerical implementation, we will replace $\pi^n$ by its un-normalized version $\tilde\pi^n$ since this simplification does not affect the shape of $\phi \mapsto K(\phi)$.  

The fact that our variational approximation is simpler than those commonly found in the literature makes it more transparent in some ways.  In particular, it is generally not clear whether a posterior approximation would inherit the statistical properties (e.g., asymptotic concentration rate) enjoyed by the posterior distribution it is approximating.  Considerable effort has been spent to prove that, in certain cases, the variational approximations do, in fact, inherit at least some of the posterior's desirable properties \citep[e.g.,][]{ray2020spike, ray2022variational,zhang2020convergence, alquier2020concentration}.  To our knowledge, one property that the aforementioned references do not establish is {\em selection consistency} of the variational approximation.  One possible reason for this gap in the existing literature is that the general tool used to transfer properties of the posterior distribution to the variational approximation requires certain exponential bounds on the posterior, behavior that have yet to be established for the selection consistency-related properties.  Since our post-marginalization approximation is simpler, we should be able to attack the problem directly and, indeed, we show below that our variational approximation $q_{\hat\phi}$ can achieve selection consistency, just like our marginal posterior $\pi^n$.

\subsection{Selection consistency}

To set the scene, let $\beta^\star$ denote the true coefficient (a $p$-vector with $p=p_n \to \infty$) and let $S^\star$ denote the true configuration.  Recall that $\beta^\star$ and $S^\star$ are actually sequences indexed by $n$, e.g., $S^\star$ can be interpreted as a binary $p_n$-vector with the only constraint being a limit on how many 1's it can contain. 
Under this setup, \citet[][Theorem~4]{tang2024empirical} showed that, under certain conditions, their empirical prior-based marginal posterior distribution for $S$ satisfies $\pi^n(S^\star) \to 1$ as $n \to \infty$ in $\prob_{\beta^\star}$-probability; note that this implies $\sum_{S \neq S^\star} \pi^n(S) \to 0$ in $\prob_{\beta^\star}$-probability.  A stronger selection consistency result under weaker conditions is established in \citet[][Corollary~6.4]{lee.chae.martin.glm}.  Since the specific conditions for model selection consistency are rather technical, and are not particularly relevant to the contributions in the present paper, we refer the reader to \citet[][Sec.~6.2]{lee.chae.martin.glm} for these details.  The one point that is relatively simple and deserves mention is that model selection consistency requires that the non-zero coefficients in the true $\beta^\star$ vector not be too small; this is the familiar {\em beta-min} condition.  Corollary~6.4 in \citet{lee.chae.martin.glm} shows that, if $|S^\star| \log p = o(n)$ and some other technical conditions (on the design matrix, prior hyperparameters, etc.)~are met, then model selection consistency holds for our marginal posterior distribution $\pi^n$ for $S$ under a beta-min condition arbitrarily close to the optimal beta-min condition, which is of the form 
\begin{equation}
\label{eq:betamin}
\min_{j \in S^\star} |\beta_j^\star| \gtrsim (n^{-1} \log p )^{1/2}.
\end{equation}
Having (nearly) the best possible beta-min condition means that model selection consistency holds under (nearly) the weakest possible conditions on $\beta^\star$.  For comparison, the corresponding beta-min condition in \citet{tang2024empirical} is strictly stronger than the near-optimal one established by \citet{lee.chae.martin.glm}.  The following theorem establishes that the proposed variational approximation of the original marginal posterior inherits the same strong selection consistency property.  

\begin{thm}
\label{thm:vb.consistent}
If the posterior distribution $\pi^n$ is model selection consistent---which it is under the conditions detailed in Section~6.2 of \citet{lee.chae.martin.glm}, including a beta-min condition arbitrarily close to that in \eqref{eq:betamin}---then the proposed variational approximation is model selection consistent too.  That is, $q_{\hat\phi}(S^\star) \to 1$ in $\prob_{\beta^\star}$-probability as $n \to \infty$.
\end{thm}

\begin{proof}
Write the objective function as $K=K_n$ to make the (previously implicit) dependence on $n$ and $p=p_n$ explicit.  We first show that $K_n(\phi)$ does not converge to $\infty$ as $n \to \infty$ for all $\phi \in [0,1]^p$.  Take $\phi^\star$ such that $\phi_j^\star = S_j^\star$, i.e., $\phi^\star$ is just $S^\star$ interpreted as a binary indicator vector.  For this choice, ``$S \sim q_{\phi^\star}$'' is a degenerate distribution at $S^\star$, with $q_{\phi^\star}(S^\star)=1$, so $K_n(\phi^\star) = -\log \pi^n(S^\star)$.  By assumption, $\pi^n(S^\star) \to 1$ in $\prob_{\beta^\star}$-probability and, consequently, $K_n(\phi^\star) \to 0$.  This implies that $K_n(\phi) \not\to \infty$ uniformly in $\phi$.  Rewriting the objective function $K_n$ as 
\begin{align*}
K_n(\phi) 
& = q_\phi(S^\star) \, \log \frac{q_\phi(S^\star)}{\pi^n(S^\star)} + \sum_{S \neq S^\star} q_\phi(S) \, \log \frac{q_\phi(S)}{\pi^n(S)},
\end{align*}
makes the following observation clear: the only way to prevent $K_n(\phi) \to \infty$---which we know is possible by the argument above---is if $q_{\phi}(S) \to 0$ for all $S \neq S^\star$.  But this implies that $q_{\phi}(S^\star) \to 1$ in $\prob_{\beta^\star}$-probability as $n \to \infty$.  The above remarks must apply to $\phi=\hat\phi$, a minimizer of $K_n$, and the desired result follows.  
\end{proof}

Note that, if $\pi^n$ was independent in the sense that the joint mass function for $S$ factored as the product of the $S_j$ marginal mass functions, then it would immediately follow that $\hat\phi_j$ equals $\pi^n(S_j = 1)$ for all $j=1,\ldots,p$.  Of course, the marginal posterior for $S$ under $\pi^n$ surely will not factor in this way, but, in the asymptotic limit, where $\pi^n$ concentrates all of its mass on $S^\star$, this factorization does hold.  Therefore, the $\hat\phi_j$ values should be close to $\pi^n(S_j=1)$, the {\em inclusion probability}.  Since the inclusion probability is what was used to carry out variable selection in the MCMC strategy used by \citet{tang2024empirical}, and since $\hat\phi_j$ is roughly approximating the inclusion probability, we can expect that the variable selection performance by our proposed variational approximation here is comparable to that of the MCMC-based method presented in \citet{tang2024empirical}---but with much faster computation.  We investigate the proximity of $\hat\phi$ to the inclusion probabilities produced by MCMC in Section~\ref{SS:eb.vi.compare}.

\subsection{Implementation}
\label{S:CAVI}

Unfortunately, directly minimizing \eqref{eq:KL} is a major challenge for non-linear models like logistic regression and other GLMs.  Note that the log-likelihood term depends on $S$ in a very complex, non-linear way, so we cannot evaluate the expected value of the log-likelihood with respect to $S \sim q_\phi$ analytically.  Alternatively, we follow \citet{ray2020spike} and introduce a new parameter $\eta$ and the following lower bound on the log-likelihood:
\begin{equation}
\ell_n(S, \beta_S)\geq \ \sum_{i=1}^n \biglb\log\tfrac{\exp(\eta_i)}{1+\exp(\eta_i)} - \tfrac{\eta_i}{2} + (y_i-\tfrac{1}{2} ) M_i(S, \beta) - \tfrac{\tanh(\eta_i/2)}{4\eta_i}[M_i(S, \beta)^2 - \eta_i^2]\bigrb,
\label{eq:logisticbound}
\end{equation}
where $M_i(S, \beta) = \sum_{j\in S} x_{ij}\beta_j$.  This lower bound on the log-likelihood leads to an upper bound on the original objective function $\phi \mapsto K(\phi)$ in \eqref{eq:KL}, so the new proposed strategy is to minimize this upper bound.  

Since $\hat{\beta}$ is the MLE, $L_n(S, \hat{\beta}_S) \geq L_n(S, \tilde{\beta}_S)$, where $\tilde{\beta}_S$ is the coefficient estimator from another method, such as lasso or SCAD. Denoting the right-hand-side of Equation~\eqref{eq:logisticbound} as $g_n(S, \beta_S)$ and plugging in Equation~\eqref{eq:laplace}, our objective function becomes
\begin{align*}
K(\phi) &= \E_{S\sim q_\phi} \Bigl\{ \log \frac{q_\phi(S)}{\pi_n(S) (1+\alpha\gamma)^{-|S|/2}L_n(S, \hat{\beta_S})^\alpha} \Bigr\} \\
&\leq \E_{S\sim q_\phi} \Bigl\{ \log q_\phi(S)-\log\pi_n(S)+\frac{|S|}{2}\log(1+\alpha\gamma)-\alpha g_n(S, \tilde{\beta}_S) \Bigr\},
\end{align*}
where $\pi_n(S)$ is the complexity prior on $S$ defined in Section~\ref{S:background}. We develop a coordinate-ascent variational inference (CAVI) algorithm to find its solution.

Our CAVI algorithm searches for the maximizer of $-K(\phi)$. Thanks to the introduction of the free parameter $\eta = (\eta_1, \eta_2, \ldots, \eta_n)^\top$, we have a closed-form equation for the lower bound approximation:
\begin{align*}
    -K(\phi) &\geq \E_{S\sim q_\phi} \Bigl\{ -\log q_\phi(S)+\log\pi_n(S)-\tfrac{|S|}{2}\log(1+\alpha\gamma)+\alpha g_n(S, \tilde{\beta}_S) \Bigr\} \\
    &= -\E_{S\sim q_\phi}\{ \log q_\phi(S) \} +\E_{S\sim q_\phi}\{ \log\pi_n(S) \} - \tfrac12 \E_{S\sim q_\phi}\{ |S| \, \log(1+\alpha\gamma) \} \\
    &\qquad + \alpha \E_{S\sim q_\phi}\{ g_n(S, \tilde{\beta}_S) \}.
\end{align*}
We will simplify each of the four components of this lower bound. The first component can be evaluated directly: 
\begin{align*}
    \E_{S\sim q_\phi} \{ \log q_\phi(S) \} &= \E_{S\sim q_\phi}\Bigl[ \sum_{j=1}^p \{ S_j\log(\phi_j)+(1-S_j)\log(1-\phi_j)\} \Bigr] \\
    & = \sum_{j=1}^p \{ \phi_j\log(\phi_j)+(1-\phi_j)\log(1-\phi_j) \}.
\end{align*} 
The third component is also straightforward, 
\begin{align*}
\E_{S\sim q_\phi} \{ |S| \, \log(1+\alpha\gamma) \} &= \log(1+\alpha\gamma) \, \sum_{j=1}^p \phi_j
\end{align*}
The second component cannot be evaluated closed-form, but the standard inequality, $\log\binom{p}{s} \leq s (1 + \log p)$, gives us a simple lower-bound:
\begin{align*}
\E_{S\sim q_\phi}\{ \log\pi_n(S)\} &= \E_{S\sim q_\phi}\{ -\log\textstyle\binom{p}{|S|}-a|S|\log p \} \\
&\geq \E_{S\sim q_\phi}\{ -|S|(1+\log p) - a|S|\log p \} \\
&= -\sum_{j=1}^p \phi_j\{ 1+\log p + a\log p \}.
\end{align*}
The fourth component requires us to evaluate the two expected values $\E_{S\sim q_\phi} M_i(S, \tilde{\beta})$ and $\E_{S\sim q_\phi}M_i(S, \tilde{\beta})^2$, where $M_i(S, \tilde{\beta}) = \sum_{j\in S} x_{ij}\tilde{\beta}_j = \sum_{j=1}^p S_jx_{ij}\tilde{\beta}_j$. This gives
\begin{align*}
    \E_{S\sim q_\phi} M_i(S, \tilde{\beta}) = \E_{S\sim q_\phi} \Bigl( \sum_{j=1}^p S_jx_{ij}\tilde{\beta}_j \Bigr) = \sum_{j=1}^p \phi_jx_{ij}\tilde{\beta}_j
\end{align*}
and 
\begin{align*}
    \E_{S\sim q_\phi} M_i(S, \tilde{\beta})^2 &= \V_{S \sim q_\phi}\{ M_i(S, \tilde{\beta}) \} + \E_{S\sim q_\phi} \{M_i(S, \tilde{\beta})\}^2\\
    & = \V_{S \sim q_\phi} \biglp\sum_{j=1}^p S_jx_{ij}\tilde{\beta}_j\bigrp + \biglp\sum_{j=1}^p \phi_jx_{ij}\tilde{\beta}_j\bigrp^2\\
    & = \sum_{j=1}^p \V_{S_j \sim {\sf Ber}(\phi_j)}(S_jx_{ij}\tilde{\beta}_j)+\biglp\sum_{j=1}^p \phi_jx_{ij}\tilde{\beta}_j\bigrp^2\\
    & = \sum_{j=1}^p\phi_j(1-\phi_j)x_{ij}^2\tilde{\beta}_j^2 +\biglp\sum_{j=1}^p \phi_jx_{ij}\tilde{\beta}_j\bigrp^2 
\end{align*}
Putting everything together, 
\begin{align*}
-K(\phi) 
&\geq -\sum_{j=1}^p\phi_j[1+\log(p)+a\log(p)-0.5\log(1+\alpha\gamma)]\\
    & \qquad +\alpha\sum_{i=1}^n\biglb\log\frac{\exp(\eta_i)}{1+\exp(\eta_i)} - \frac{\eta_i}{2} + (y_i-\tfrac12) \sum_{j=1}^p\phi_j x_{ij}\tilde{\beta}_j \\
    & \qquad \qquad \qquad - \frac{1}{4\eta_i}\tanh{\frac{\eta_i}{2}}\bigls\sum_{j=1}^p\phi_j(1-\phi_j)x_{ij}^2\tilde{\beta}_j^2 +\biglp\sum_{j=1}^p \phi_jx_{ij}\tilde{\beta}_j\bigrp^2 - \eta_i^2\bigrp\bigrs\bigrb\\
    & \qquad -\sum_{j=1}^p \bigls\phi_j\log(\phi_j)+(1-\phi_j)\log(1-\phi_j)\bigrs.
\end{align*}

Taking the derivative of the above with respect to $\phi_j$, and setting $\omega_j = \log(\frac{\phi_j}{1-\phi_j})$, we have the following update equation for the $(t+1)^\text{st}$ iteration of $\phi_j$,
\begin{align}
    \omega_j^{(t+1)} & = \alpha\tilde{\beta}_j\sum_{i=1}^n ( y_i-\tfrac12) x_{ij} - \frac{\alpha\tilde{\beta}_j}{4}\sum_{i=1}^n\frac{1}{\eta^{(t)}_i}\tanh\biglp\frac{\eta^{(t)}_i}{2}\bigrp \biglp x_{ij}^2\tilde{\beta}_j+2x_{ij}\sum_{k\neq j}\phi^{(t)}_kx_{ik}\tilde{\beta}_k\bigrp \notag \\
    & \qquad +\tfrac{1}{2}\log(1+\alpha\gamma)-(a+1)\log(p)-1
    \label{eq:omega}
\end{align}
\begin{equation}
    \phi_j^{(t+1)} = \frac{\exp(\omega_j^{(t+1)})}{1+\exp(\omega^{(t+1)}_j)}
    \label{eq:phi}
\end{equation}
Our free parameter $\eta$ is updated with 
\begin{equation}
    \eta_i^{(t+1)} = \E_q^{1/2}\{ M_i(S)^2\} = \Bigl\{ \sum_{j=1}^p\phi^{(t)}_j(1-\phi^{(t)}_j)x_{ij}^2\tilde{\beta}_j^2 +\biglp\sum_{j=1}^p \phi^{(t)}_jx_{ij}\tilde{\beta}_j\bigrp^2 \Bigr\}^{1/2}
    \label{eq:eta}
\end{equation}

We provide details of our CAVI algorithm in Algorithm~\ref{algo}. Inputs for our algorithm include the data $X$ and observables $y$, a fixed estimator $\tilde{\beta}$, an initial value for $\phi$, a stopping threshold $\eps$, and a maximum number of iterations to ensure the algorithm would stop even without converging. The fixed estimator could be based on any other reliable method---both lasso and SCAD work fairly well---that is also not difficult to compute. In practice, we took our $\tilde{\beta}$ from SCAD, with a small caveat---the form of Equation~\eqref{eq:omega} necessitates that zero entries in $\tilde{\beta}$ will not be updated iteratively, so we take the zero entries from SCAD and add a small amount of noise to them. The initial value for $\phi$ could be also set to match $\tilde{\beta}$, but in practice, we have found that $\phi^{(0)} = (0.5, \ldots, 0.5)^\top$ works equally well. 

\begin{algorithm}[t]
	\caption{CAVI for variational empirical Bayes} 
 Input: data $(X, y)$; a fixed estimator $\tilde{\beta}$ based on another method such as SCAD or Lasso; an initial $\phi^{(0)}$ for the initial value of $\phi$; a stopping threshold $\epsilon$; a maximum number of iterations max.iter.
	\begin{algorithmic}[1]
    \State Initialize $\eta$ with Equation~\eqref{eq:eta} and $\phi^{(0)}$.
    \State Calculate $\omega^{(0)}$, $\omega^{(1)}$ by Equation~\eqref{eq:omega} and $\phi^{(1)}$ based on Equation~\eqref{eq:phi}.
    \State Set $t=1$.
		\While {$\max_j |H(\phi^{(t)}_j) - H(\phi^{(t-1)}_j)|>\eps$ and $t<$ max.iter} 
			\For {$j=1,2,\ldots,p$}
				\State Update $\omega^{(t+1)}_j$ by Equation~\eqref{eq:omega}
				\State Compute $\phi^{(t+1)}_j$ based on Equation~\eqref{eq:phi}
			\EndFor
			\State Update $\eta$ with Equation~\eqref{eq:eta}
            \State $t = t+1$
		\EndWhile
  		\State Return $\hat{\phi} = \phi^{(t)}$
	\end{algorithmic} 
 \label{algo}
\end{algorithm}

The stopping criterion evaluates the difference of $\phi$ values between consecutive iterations, and the algorithm stops when the difference is below a threshold $\eps$, specifically, following \citet{ray2020spike}, \citet{yang2020variational}, and \citet{huang2016variational}, we look at the maximum entropy criterion. We stop our algorithm when $\max_j |H(\phi^{(t)}_j) - H(\phi^{(t-1)}_j)|>\eps$, where $H: [0,1] \to \RR$ is defined as $H(z) = -z\log_2(z)-(1-z)\log_2(1-z)$.

For variable selection, the solution $\hat{\phi}$ from Algorithm~\ref{algo} offers an approximation of inclusion probabilities for each $\beta_j$ to be included in the true active set $S^\star$. For point estimation, we can then take the indices $j$ that satisfy $\hat{\phi}_j\geq 0.5$, and calculate the MLE $\hat{\beta}$ based on the model with the chosen set of indices.

\section{Results}
\label{S:chap4results}
We demonstrate the efficiency and accuracy of our method with numerical simulations. We compare our method, EB-VI, with a number of state-of-the-art methods for high-dimensional logistic regression. We first compare our method with other Bayesian methods including variational Bayes with Laplace prior \citep{ray2020spike} and with Gaussian prior \citep[e.g.][]{huang2016variational, ormerod.etal.2017}, and R packages varbvs \citep{carbonetto2012scalable}, SkinnyGibbs \citep{narisetty2018skinny}, and BinaryEMVS \citep{mcdermott2016methods}. We also compare our method to a few other popular frequentist and Bayesian methods, including horseshoe, lasso, adaptive lasso \citep{zou2006}, SCAD \citep{fanli2001}, and MCP \citep{chzhang2010}. Lastly, we compare our EB-VI to the MCMC method presented in detail in \citet{tang2024empirical}, denoted here as EB-MCMC, especially looking at the difference between the solution $\hat{\phi}$ obtained in EB-VI and the inclusion probabilities from EB-MCMC, to investigate whether our variational approximation effectively solves the original problem of interest. 

\subsection{Comparisons with other methods}
For the comparisons with other state-of-the-art Bayesian methods, we consider the same simulation settings as \citet{ray2020spike}. We look at five different simulation settings, with 200 runs each. In all five test settings, the entries of the design matrix $X$ are simulated from an independent Normal distribution with mean 0 and varying standard deviation $\sigma$, i.e., $X_{ij}\sim \N(0, \sigma^2)$. The number of active signals $s$ and the size of the signals $A$ both vary in the five tests. Tests 1--3 take $n=250$, $p=500$, and Tests 4--5 are higher-dimensional, with $n=2500$, $p=5000$. All tests place the nonzero signals at the beginning of the true coefficient vector. The specific settings are summarized below: 
\begin{description}
    \item[Test 1.] $n=250$, $p=500$, $\sigma=0.25$, $s=5$, and $A=4$
    \item[Test 2.] $n=250$, $p=500$, $\sigma=2$, $s=10$, and $A=6$
    \item[Test 3.] $n=250$, $p=500$, $\sigma=0.5$, $s=15$, and $A\sim \text{Unif }(-2,2)$
    \item[Test 4.] $n=2500$, $p=5000$, $\sigma=0.5$, $s=25$, and $A=2$
    \item[Test 5.] $n=2500$, $p=5000$, $\sigma=1$, $s=10$, and $A\sim \text{Unif } (-1,1)$
\end{description}

We compare two metrics, True Positive Rate (TPR) and False Discovery Rate (FDR), defined as, respectively, 
\[\text{TPR} = \frac{\text{TP}}{\text{TP}+\text{FN}} \quad \text{and} \quad \text{FDR} = \frac{\text{FP}}{\text{TP}+\text{FP}},\]
where TP is the number of truely positive signals identified by the algorithm, FN is the number of signals that the algorithm failed to identify, and FP is the number of variables that are actually noise but falsely selected by the algorithm. 

The results are shown in Table~\ref{table:EBVI1}, where rows 2 through 6 of both the TPR and FDR values are taken from Table 3 of \citet{ray2020spike}, and row 1 is from our newly proposed method. Our method does well in locating the signals in Tests 1, 2, and 4, as evidenced by the high TPR values. Tests 3 and 5 have the most challenging settings for signal discovery, as the TPR values are the lowest for all methods. In terms of FDR, our method does very well, obtaining very low FDR values in all five tests and outperforming the other methods in four of the five. Our EB-VI method is also very efficient, finishing each single run in around a second for the first three settings with $n=250$, $p=500$, on-par with the time required for the variational Bayesian with Laplace prior (VB-Laplace) method of ~\citet{ray2020spike}; for the higher-dimensional tests 4 and 5, our method generally took 2--3 minutes for a single run, slightly slower than the VB-Laplace method. 

\begin{table}[t]
\caption{Comparison of TPR and FDR for select Bayesian methods across five test settings. Rows 2--6 of both the TPR and FDR panels are taken from Table 3 in \citet{ray2020spike}. }
{\small
\begin{center}  
\begin{tabular}{cccccccccccc}
    \hline
     & Algorithm & Test 1 & Test 2 & Test 3 & Test 4 & Test 5\\
     \hline
         \multirow{6}{*}{TPR} & \textbf{EB-VI}& $0.96\pm0.10$ & $1.00\pm 0.00 $ & $0.30\pm 0.12 $ & $1.00\pm0.00$ & $0.31\pm0.09$\\
        & VB(Lap) & $0.99\pm0.06$ & $1.00\pm0.00$ &	$0.51\pm0.11$ &	$1.00\pm0.00$ &	$0.40\pm0.28$ \\
         & VB(Gauss) & $1.00\pm 0.01$&  $1.00\pm0.02$ & $0.54\pm0.11$ &$1.00\pm0.00$	&$0.85\pm0.06$ \\
         & varbvs& $1.00\pm0.00$ & $1.00\pm0.00$  &$0.68\pm0.11$	&$1.00\pm0.00$	&$0.87\pm0.06$ \\
         & SkinnyGibbs& $0.98\pm0.06$ &$1.00\pm0.02$ &$0.51\pm0.12$	& -- & -- \\
        & BinEMVS& $0.99\pm0.03$& $1.00\pm0.00$  &$0.58\pm0.11$ &-- & -- \\[3pt]
    \hline
        \multirow{6}{*}{FDR} & \textbf{EB-VI}& $0.03\pm0.08$& $0.03\pm0.05$ & $0.03\pm0.08$ & $0.00\pm0.01$ & $0.01\pm0.07$\\
        &  VB(Lap) & $0.49\pm0.11$ & $0.00\pm0.02$&	$0.41\pm0.14$ &	$0.01\pm0.02$&	$0.03\pm0.05$ \\
         & VB(Gauss) & $0.63\pm0.07$& $0.09\pm0.13$  &$0.52\pm0.12$	&$0.81\pm0.02$	&$0.95\pm0.01$ \\
         & varbvs&$0.93\pm0.01$ & $0.08\pm0.08$  &$0.83\pm0.03$	&$0.93\pm0.00$	&$0.91\pm0.01$ \\
         & SkinnyGibbs& $0.80\pm0.03$& $0.11\pm0.11$  &$0.71\pm0.07$	& -- & -- \\
        & BinEMVS& $0.43\pm0.14$ & $0.19\pm0.10$  &$0.63\pm0.10$ & -- & -- \\[3pt]
    \hline
\end{tabular}
\end{center}
}
\label{table:EBVI1}
\end{table}

We also compare EB-VI against a number of other methods discussed in \citet{tang2024empirical} under their same simulation settings. The methods include lasso, adaptive lasso, SCAD, MCP, horseshoe, and skinnyGibbs. We compare these methods using the follow metrics: sensitivity (TPR), specificity (TNR), and Matthews correlation coefficient (MCC). TPR is as defined above, and TNR and MCC are defined as
\[\text{TNR} = \frac{\text{TN}}{\text{TN}+\text{FP}}\] and
\[\text{MCC} = \frac{\text{TP}\times \text{TN} - \text{FP}\times \text{FN}}{\{(\text{TP}+\text{FP})(\text{TP}+\text{FN})(\text{TN}+\text{FP})(\text{TN}+\text{FN})\}^{1/2}},\] where TP, FP, and FN are as defined previously, and TN is the number of noise variables that are correctly categorized by the algorithm. 
For these simulations, we fixed sample size $n=100$ and varied $p \in \{200, 400\}$. The true $\beta^\star$ has its nonzero components set to 3, with the number of nonzero components $s$ varying between 4 and 8. The design matrix $X$ has its rows randomly drawn from a multivariate Gaussian distribution with mean 0, variance 1, and covariance matrix $\Sigma$, where $\Sigma_{ij} = r^{|i-j|}$, a first-order autoregressive correlation structure, with varied $r \in \{0, 0.2\}$. We draw independent response variables, with $y_i\sim \mathsf{Ber}(\{1 + \exp(-x_i^\top\beta^\star)\}^{-1})$, for $i = 1,\ldots,n$. We run 500 replications at each of the eight settings, and the results are summarized in Table~\ref{table:EBVI2}.


The right-most column includes results for our new method, EB-VI, while the other columns are taken from Table 1 of \citet{tang2024empirical}. We see that EB-VI does fairly well across these settings for all three metrics. Our method has very high True Negative Rate (TNR) across all settings and in three out of the eight total settings, EB-VI has the highest MCC value. As the sparsity level decreases (the settings with $|S|=8$), it is harder to identify all the signals, we see the TPR values for our method suffer a little, but is still on-par with other methods, and generally similar to the TPR values from the EB-MCMC method (denoted EB1 and EB2 in Table~\ref{table:EBVI2}).

\begin{table}[t]
\caption{Comparison of TPR, TNR, and MCC for select frequentist and Bayesian methods across various settings. Columns 1--8 are taken from Table 1 of \citet{tang2024empirical}. Column 9 are results from the newly-proposed EB-VI. }
{\footnotesize
\begin{center}  
\begin{tabular}{ccccccccccccc}
    \hline
    $p$ & $|S|$ & $r$ & Metric & EB1 & EB2 & HS & lasso & alasso & SCAD & MCP & skinny & {\bf EB-VI} \\
     \hline
         \multirow{3}{*}{200} &  \multirow{3}{*}{4} & \multirow{3}{*}{0} & TPR & 0.980 & 0.966 &	0.871& 0.998	&	0.701&	1.000	&0.999 & 0.997 & 0.990\\
         & & &  TNR &0.971& 0.992 &1.000&	0.953&	0.854 &0.962 &	0.986 &0.999 & 0.999 \\
          & & & MCC & 0.676	&0.859&	0.927	&0.630	&0.228	&0.597 &0.783 &0.980 & 0.976\\[3pt]
             \multirow{3}{*}{200} &  \multirow{3}{*}{4} & \multirow{3}{*}{0.2}& TPR &  0.962 & 0.926 &	0.786&	0.995 &    0.806&	0.999 & 0.996 & 0.985 & 0.995\\
         &  &  &TNR & 0.975& 0.994  &	1.000&	0.966 &0.847& 0.955& 0.983 & 0.999 & 0.999 \\
          &  & & MCC &  0.691 &0.856 &	0.878&	0.708&  0.272& 0.561& 0.744& 0.973& 0.979\\[3pt]
             \multirow{3}{*}{200} &  \multirow{3}{*}{8} & \multirow{3}{*}{0}& TPR &  0.746 & 0.643 &	0.276&	0.959 & 0.579&	0.962&	0.925 & 0.813 & 0.649\\
         &  &  &TNR & 0.940	&0.983&	1.000&	0.895&	0.851&	0.948&	0.982 & 0.998 & 0.999\\
          &  & & MCC &  0.519 & 0.618&	0.498&	0.528 &  0.206& 0.640& 0.790& 0.865 & 0.780\\[3pt]
             \multirow{3}{*}{200} &  \multirow{3}{*}{8} & \multirow{3}{*}{0.2}& TPR &  0.747& 0.648 &0.287	&0.953&	0.705&	0.951&	0.888 &0.714 & 0.655\\
         &  &  &TNR & 0.956&0.988&	1.000&	0.932& 0.837& 0.952& 0.984&0.998 & 0.999\\
          &  & & MCC & 0.570& 0.663& 0.514& 0.631&  0.268&0.644 &0.780& 0.806& 0.788\\[3pt]
             \multirow{3}{*}{400} &  \multirow{3}{*}{4} & \multirow{3}{*}{0}& TPR & 0.922&	0.897& 0.582&	0.993& 0.741& 1.000&	1.000 &0.980 & 0.990\\
         &  &  &TNR & 0.989&	0.996&	1.000	&0.973&	0.924	&0.974&	0.991 & 0.998 & 1.000\\
          &  & & MCC &  0.728	&0.841&	0.728	&0.618&	0.265&0.543&	0.744 & 0.920 & 0.983\\[3pt]
         \multirow{3}{*}{400} &  \multirow{3}{*}{4} & \multirow{3}{*}{0.2} & TPR & 0.926&0.860	&0.568	&0.994&	0.836&	0.998	&0.995 & 0.952 & 0.983\\
         & & &  TNR & 0.993	&0.998	&1.000&	0.980&	0.931	&0.971	&0.990 &0.998 & 1.000 \\
          & & & MCC & 0.795&	0.860&	0.737	&0.685	&0.336	&0.514	&0.718 & 0.899& 0.983\\[3pt]
             \multirow{3}{*}{400} &  \multirow{3}{*}{8} & \multirow{3}{*}{0}& TPR &  0.490&	0.407&	0.057	&0.896&	0.468	&0.917	&0.843 & 0.498 & 0.449 \\
         &  &  &TNR & 0.974&	0.987&	1.000	&0.941	&0.930	&0.962	&0.987 & 0.996 & 1.000 \\
          &  & & MCC & 0.406	&0.401	&0.145&	0.493&	0.196&	0.546	&0.694 & 0.591 &0.617\\[3pt]
             \multirow{3}{*}{400} &  \multirow{3}{*}{8} & \multirow{3}{*}{0.2}& TPR &  0.548&	0.438&	0.092&	0.931	&0.617&	0.922	&0.842 & 0.565 &0.485 \\
         &  &  &TNR & 0.983& 0.992 &1.000&	0.958	&0.921	&0.965&	0.989 &0.996 &1.000 \\
          &  & & MCC &  0.513&	0.501	&0.226&	0.580	&0.260	&0.564&	0.714 &0.646 &0.674\\
    \hline
\end{tabular}
\end{center}
}
\label{table:EBVI2}
\end{table}

\subsection{Comparisons with EB-MCMC}
\label{SS:eb.vi.compare}

We discussed the theoretical connections in Section~\ref{S:VI} between the proposed EB-VI solution in this paper and the EB-MCMC solution from \citet{tang2024empirical}.  Here, we demonstrate empirically that the relevant features of our newly proposed EB-VI method well-approximate the MCMC inclusion probabilities. To compare the $\hat{\phi}$ against the inclusion probabilities, denoted by $\hat{\pi}^n$, we will look at the following metric,
\begin{equation}
    D = \bigl( p^{-1} \, \E \|\hat{\pi}^n-\hat{\phi}\|^2 \bigr)^{1/2},
\label{eq:D}
\end{equation}
where the expectation is with respect to data. Note that the $\hat{\phi}$ here is not exactly the minimizer to the Kullback--Leibler divergence of $\pi^n$ from $q_\phi$, but the solution to the CAVI algorithm detailed in Algorithm~\ref{algo}. $\hat{\pi}^n$ is the inclusion probabilities for the empirical priors posterior distribution, as approximated by MCMC with $M=10,000$ samples. 

Table~\ref{table:VIvsMCMC} provides a comparison of the distance $D$ across select settings, each setting with 100 runs. The true coefficient vector $\beta^\star = (A, \ldots, A, 0, \ldots, 0)^\top$, with signal size $A$ and the number of true signals $s$. The entries $X_{ij}$'s in the design matrix $X$ are independent Gaussians, i.e., $X_{ij} \iid \N(0,1)$. As is evident from Table~\ref{table:VIvsMCMC}, the distances between $\hat{\phi}$ and $\hat{\pi}^n$ are small across the different settings tested, which is another indication that numerically, our CAVI algorithm is approximating the inclusion probabilities from MCMC well. In terms of runtime, as one would expect, EB-VI is much faster than EB-MCMC. For a dataset that takes EB-MCMC around 30 seconds to run, it would take EB-VI roughly between 0.05 to 0.1 seconds. 

\begin{table}[t]
\caption{The distance $D$ as defined in Equation~\eqref{eq:D} between the EB-VI solution and the EB-MCMC inclusion probabilities for various settings. }
\begin{center}  
\begin{tabular}{c|ccccccccccc}
    \hline
     $(n, p, s, A)$ & (100, 200, 4, 3) & (100, 200, 6, 3) & (100, 200, 4, 6) & (200, 400, 4, 3) &  \\
     \hline
       $D$   & 0.080 & 0.098 & 0.046 & 0.052\\
    \hline
\end{tabular}
\end{center}
\label{table:VIvsMCMC}
\end{table}

\section{Discussion} 
\label{S:chap4discuss}
In this paper, we propose a novel variational approximation directly on the marginal posterior for $S$ for variable selection in high-dimensional logistic regression. Following \citet{tang2024empirical} and \citet{lee.chae.martin.glm}, we use an empirically-centered prior, which yields a marginal posterior $S$ that has a simple form after Laplace approximation. Our simple independent-Bernoulli approximation for $S$ shrinks the variational parameter space from what is typically $3p$ to just $[0,1]^p$, streamlining the CAVI algorithm for computations.  Thanks to its relative simplicity, we can easily prove (Theorem~\ref{thm:vb.consistent}) that our proposed variational approximation shares the same strong selection consistency property as the marginal posterior it is approximating. Simulations show that our method is efficient and produces good results, on par with existing state-of-the-art methods. 

We explored the relationship between the solution of our variational approximation and inclusion probabilities obtained with MCMC through simulations, and confirmed that the variational method does well in approximating the posterior distribution empirically. Theoretical properties of the relationship between the variational approximation and the original posterior of interest has yet to be explored, and would be an area of focus for further research. 
Finally, the proposed methodology can be extended beyond the logistic regression case to approximate the marginal posterior distribution for other high-dimensional GLMs; the challenge is finding a counterpart to the lower bound \eqref{eq:logisticbound} above for the log-likelihood for other GLMs.

\bibliographystyle{apalike}
\bibliography{mybib.bib}

\end{document}